\documentclass[10pt,reqno]{amsart}
\usepackage[hypertex]{hyperref}

\usepackage{amsmath,amsthm}
\usepackage{amssymb}
\usepackage{euscript}
\usepackage[frame,ps,matrix,arrow,curve,rotate,all,2cell,tips]{xy}

\numberwithin{equation}{subsection}

\newcommand{\sqsp}{\renewcommand{\baselinestretch}{1.1}\tiny\normalsize}

\newtheorem{theorem}[subsection]{Theorem}
\newtheorem{lemma}[subsection]{Lemma}
\newtheorem{proposition}[subsection]{Proposition}
\newtheorem{corollary}[subsection]{Corollary}

\theoremstyle{definition}
\newtheorem{definition}[subsection]{Definition}
\newtheorem{example}[subsection]{Example}
\newtheorem{remark}[subsection]{Remark}

\newcommand{\fg}{\mathfrak{g}}

\newcommand{\ba}{\mathbf{a}}
\newcommand{\bc}{\mathbf{c}}
\newcommand{\bC}{\mathbf{C}}
\newcommand{\bk}{\mathbf{k}}
\newcommand{\bN}{\mathbf{N}}

\newcommand{\bZ}{\mathbf{Z}}

\def\uhg{U_h(\fg)}
\def\uhsl{U_h(sl_2)}
\def\ch{\bC[[h]]}
\def\an{A^{(n)}}
\def\rn{R^{\alpha^n}}
\def\rt{R^{\alpha^t}}
\def\ralpha{R^{\alpha}}
\def\zn{\bZ'_{/n}}
\def\kg{\bk G}
\def\kofg{\bk(G)}
\def\hc{H_{\bc}}
\def\alphac{\alpha_{\bc}}
\def\vn{\widetilde{V}_n}
\def\vone{\widetilde{V}_1}

\DeclareMathOperator{\Hom}{Hom}

\DeclareMathOperator{\diag}{diag}

\def\Bmatrix{{\begin{pmatrix}q & 0 & 0 & 0\\
0 & 0 & \gamma^{-1} & 0\\
0 & \gamma^{-1} & \gamma^{-1}(q - q^{-1}) & 0 \\
0 & 0 & 0 & q\gamma^{-2}\end{pmatrix}}}

\def\Brmatrix{{\begin{pmatrix}q & 0 & 0 & 0\\
0 & 0 & 1 & 0\\
0 & 1 & q - q^{-1} & 0 \\
0 & 0 & 0 & q\end{pmatrix}}}

\begin{document}

\title{Hom-quantum groups I: quasi-triangular Hom-bialgebras}
\author{Donald Yau}

\begin{abstract}
We introduce a Hom-type generalization of quantum groups, called quasi-triangular Hom-bialgebras.  They are non-associative and non-coassociative analogues of Drinfel'd's quasi-triangular bialgebras, in which the non-(co)associativity is controlled by a twisting map.  A family of quasi-triangular Hom-bialgebras can be constructed from any quasi-triangular bialgebra, such as Drinfel'd's quantum enveloping algebras.  Each quasi-triangular Hom-bialgebra comes with a solution of the quantum Hom-Yang-Baxter equation, which is a non-associative version of the quantum Yang-Baxter equation.  Solutions of the Hom-Yang-Baxter equation can be obtained from modules of suitable quasi-triangular Hom-bialgebras.
\end{abstract}



\address{Department of Mathematics\\
    The Ohio State University at Newark\\
    1179 University Drive\\
    Newark, OH 43055, USA}
\email{dyau@math.ohio-state.edu}

\date{\today}
\maketitle

\sqsp

\section{Introduction}


This paper is part of an on-going effort \cite{yau5,yau6,yau7} to study twisted, Hom-type generalizations of the various Yang-Baxter equations and related algebraic structures.  A Hom-type generalization of the Yang-Baxter equation \cite{baxter,baxter2,yang}, called the Hom-Yang-Baxter equation (HYBE), and its relationships to the braid relations and braid group representations \cite{artin2,artin} were studied in \cite{yau5,yau6}.  Hom versions of the classical Yang-Baxter equation \cite{skl1,skl2} and Drinfel'd's Lie bialgebras \cite{dri83,dri87} were studied in \cite{yau7}.

Here we consider twisted, Hom-type generalizations of quantum groups and the quantum Yang-Baxter equation (QYBE).  The quantum groups being generalized in this paper are Drinfel'd's quasi-triangular bialgebras \cite{dri87}.  Our generalized quantum groups (the quasi-triangular Hom-bialgebras in the title) are, in general, non-associative, non-coassociative, non-commutative, and non-cocommutative.  We also refer to these objects colloquially as \emph{Hom-quantum groups}.  As we will describe below, suitable quasi-triangular Hom-bialgebras give rise to solutions of the HYBE.

Let us first recall the definition of a quasi-triangular bialgebra and its relationships to the various Yang-Baxter equations.  A \emph{quasi-triangular bialgebra} $(A,R)$ \cite{dri87} consists of a bialgebra $A$ and an invertible element $R \in A^{\otimes 2}$ such that the following three conditions are satisfied:
\begin{equation}
\label{qtb}
\begin{split}
\Delta^{op}(x)R &= R\Delta(x),\\
(\Delta \otimes Id)(R) &= R_{13}R_{23}, \quad (Id \otimes \Delta)(R) = R_{13}R_{12}.
\end{split}
\end{equation}
Here $\Delta^{op} = \tau \circ \Delta$ with $\tau$ the twist isomorphism, $R_{12} = R \otimes 1$, $R_{23} = 1 \otimes R$, and $R_{13} = (\tau \otimes Id)R_{23}$.  The element $R$, called the \emph{quasi-triangular structure}, satisfies the QYBE
\begin{equation}
\label{qybe}
R_{12}R_{13}R_{23} = R_{23}R_{13}R_{12}.
\end{equation}
Examples of quasi-triangular bialgebras include Drinfel'd's quantum enveloping algebra $\uhg$ \cite{dri87} of a semi-simple Lie algebra or a Kac-Moody algebra $\fg$ \cite{kac}, the anyonic quantum groups \cite{majid92}, and the quantum line \cite{lm}, among many others.  The QYBE and quasi-triangular bialgebras are motivated by work on the quantum inverse scattering method \cite{fad80,fad,fad0,fad1,fad2,skl1,skl2} and exactly solved models in statistical mechanics \cite{baxter,baxter2,yang}.  Comprehensive expositions on quasi-triangular bialgebras can be found, for example, in the books \cite{cp,es,kassel,majid}.

Quasi-triangular bialgebras and the QYBE are related to the Yang-Baxter equation as follows.  Consider a module $V$ over a quasi-triangular bialgebra $(A,R)$ and the operator $B \colon V^{\otimes 2} \to V^{\otimes 2}$ defined as $B(v \otimes w) = \tau(R(v \otimes w))$.  As a consequence of the QYBE \eqref{qybe}, the operator $B$ satisfies the \emph{Yang-Baxter equation} (YBE) \cite{baxter,baxter2,yang},
\begin{equation}
\label{ybe}
(B \otimes Id)(Id \otimes B)(B \otimes Id) = (Id \otimes B)(B \otimes Id)(Id \otimes B).
\end{equation}
So one quasi-triangular bialgebra $(A,R)$ gives rise to many solutions of the YBE via its modules.  For this reason, the quasi-triangular structure $R$ is also known as the \emph{universal $R$-matrix}.  The reader may consult \cite{perk} for discussion of different versions of the YBE and of their uses in physics.


Our generalizations of quantum groups and the QYBE (as well as the HYBE) are all motivated by Hom-Lie algebras and other Hom-type algebras.  Roughly speaking, a Hom-type structure arises when one strategically replaces the identity map in the defining axioms of a classical structure by a general twisting map $\alpha$.  A classical structure should be a particular example of a Hom-type structure in which the twisting map is the identity map.  In particular, a \emph{Hom-Lie algebra} $(L,[-,-],\alpha)$ has an anti-symmetric bracket $[-,-] \colon L^{\otimes 2} \to L$ that satisfies the \emph{Hom-Jacobi identity},
\[
[[x,y],\alpha(z)] + [[z,x],\alpha(y)] + [[y,z],\alpha(x)] = 0,
\]
where $\alpha$ is an algebra self-map of the module $L$.  Hom-Lie algebras were introduced in \cite{hls} to describe the structures on some $q$-deformations of the Witt and the Virasoro algebras.  Earlier precursors of Hom-Lie algebras can be found in \cite{hu,liu}.  Hom-Lie algebras are closely related to deformed vector fields \cite{ama,hls,ls,ls2,ls3,rs,ss} and number theory \cite{larsson}.

One can similarly define a \emph{Hom-associative algebra} \cite{ms}, which satisfies the \emph{Hom-associativity axiom} $(xy)\alpha(z) = \alpha(x)(yz)$ (Definition ~\ref{def:homas}).  So a Hom-associative algebra is not associative, but the non-associativity is controlled by the twisting map $\alpha$.  One obtains a Hom-Lie algebra from a Hom-associative algebra via the commutator bracket \cite[Proposition 1.7]{ms}.  Conversely, the enveloping Hom-associative algebra of a Hom-Lie algebra is constructed in \cite{yau} and is studied further in \cite{yau3}.  Other papers concerning Hom-Lie algebras and related Hom-type structures are \cite{atm,fg,fg2,gohr,jl,ms2,ms3,ms4,yau,yau2,yau4,yau5,yau6,yau7}.

We now describe the main results of this paper concerning Hom-type generalizations of quasi-triangular bialgebras and the QYBE.  Precise definitions, statements of results, and proofs are given in later sections.

Following the patterns of Hom-Lie and Hom-associative algebras, one can define Hom-bialgebras (Definition ~\ref{def:homas}), which are non-associative and non-coassociative generalizations of bialgebras in which the non-(co)associativity is controlled by the twisting map $\alpha$.  In section ~\ref{sec:qthb} we introduce \emph{quasi-triangular Hom-bialgebras}(Definition ~\ref{def:qthb}), generalizing Drinfel'd's quasi-triangular bialgebras by strategically replacing the identity map by a twisting map $\alpha$ in the defining axioms.  We show that a quasi-triangular Hom-bialgebra comes equipped with a solution $R$ of the \emph{quantum Hom-Yang-Baxter equation} (QHYBE) (Theorem ~\ref{thm:qhybe}), which is a non-associative analogue of the QYBE.  In fact, due to the non-associative nature of a Hom-bialgebra, there are two different versions of the QHYBE (\eqref{qhybe} and \eqref{qhybe'}), both of which hold in a quasi-triangular Hom-bialgebra.

In section ~\ref{sec:example} we give two general procedures by which quasi-triangular Hom-bialgebras can be constructed.  First we show that every quasi-triangular bialgebra $A$ can be twisted into a family of quasi-triangular Hom-bialgebras $A_\alpha$, where $\alpha$ runs through the bialgebra endomorphisms on $A$.  Here the twisting procedure is applied to the (co)multiplication in $A$ (Theorem ~\ref{thm:twist1}).  On the other hand, if $A$ is a quasi-triangular Hom-bialgebra with a surjective twisting map $\alpha$, then we obtain a sequence of quasi-triangular Hom-bialgebras by replacing $R \in A^{\otimes 2}$ with $(\alpha^n \otimes \alpha^n)(R)$ for $n \geq 1$ (Theorem ~\ref{thm:twist2}).  These twisting procedures yield lots of examples of quasi-triangular Hom-bialgebras.  As illustrations, we apply these twisting procedures to Majid's anyon-generating quantum groups \cite{majid92} (Example ~\ref{ex:anyon}), the quasi-triangular group bialgebra $\bk G$ (Example ~\ref{ex:groupbialg}) and function bialgebra $\bk(G)$ (Example ~\ref{ex:function}) for finite abelian groups $G$, and Drinfel'd's quantum enveloping algebras $\uhg$ (Examples ~\ref{ex:env} and ~\ref{ex:sl2}) for semi-simple Lie algebras $\fg$.  Since $\uhg$ is non-commutative and non-cocommutative, the quasi-triangular Hom-bialgebras obtained by twisting $\uhg$ are simultaneously non-(co)associative and non-(co)commutative.

The relationship between the QYBE \eqref{qybe} and the YBE \eqref{ybe} described above is generalized to the Hom setting in section ~\ref{sec:hybe}.  We show that a module (suitably defined) over a quasi-triangular Hom-bialgebra with $\alpha$-invariant $R$ has a canonical solution of the HYBE given by $B = \tau \circ R$ (Theorem ~\ref{thm:hybe}).  This generalizes the solution of the YBE associated to a module over a quasi-triangular bialgebra, as discussed above.

We illustrate the concept of modules over a quasi-triangular Hom-bialgebra with $\alpha$-invariant $R$ in section ~\ref{sec:module}.  In particular, we consider the quasi-triangular Hom-bialgebra $\uhsl_\alpha$ (Example ~\ref{ex:sl2}) and a sequence of modules $\vn$ over it (Proposition ~\ref{vnmodule}).  Here $\vn$ is topologically free of rank $n+1$ over $\ch$.  We also write down explicitly the matrix representing the canonical solution of the HYBE associated to $\vone$ (Proposition ~\ref{BalphaV1}).

In the next paper in this series, we will discuss another class of Hom-quantum groups, called dual quasi-triangular (DQT) Hom-bialgebras, which are also non-(co)associative and non-(co)commutative in general.  These DQT Hom-bialgebras are Hom-type generalizations of dual quasi-triangular bialgebras \cite{hay,lt,majid91,sch}.  In particular, they can accommodate Hom-type generalizations of the FRT quantum groups \cite{rft}, quantum matrices $M_q(2)$, the quantum general linear group $GL_q(2)$, and the quantum special linear group $SL_q(2)$.  Applied to the $M_q(2)$-coaction on the quantum plane, we obtain Hom-type, non-(co)associative analogues of quantum/non-commutative geometry.  DQT Hom-bialgebras are also equipped with solutions of a Hom version of the operator form of the QYBE.  Solutions of the HYBE can be generated from comodules of suitable DQT Hom-bialgebras.

\section{Quasi-triangular Hom-bialgebras}
\label{sec:qthb}

In this section, we first recall the definition of a Hom-bialgebra (Definition ~\ref{def:homas}).  Then we define quasi-triangular Hom-bialgebras (Definition ~\ref{def:qthb}) and establish the QHYBE (Theorem ~\ref{thm:qhybe}).  Several characterizations of the axioms of a quasi-triangular Hom-bialgebra are given at the end of this section (Theorems ~\ref{thm:axiom1} and \ref{thm:axiom2}).  Concrete examples of quasi-triangular Hom-bialgebras are given in the next section.

\subsection{Conventions and notations}

We work over a fixed commutative ring $\bk$ of characteristic $0$.  Modules, tensor products, and linear maps are all taken over $\bk$.  If $V$ and $W$ are $\bk$-modules, then $\tau \colon V \otimes W \to W \otimes V$ denotes the twist isomorphism, $\tau(v \otimes w) = w \otimes v$.  For a map $\phi \colon V \to W$ and $v \in V$, we sometimes write $\phi(v)$ as $\langle \phi,v\rangle$.  If $\bk$ is a field and $V$ is a $\bk$-vector space, then the linear dual of $V$ is $V^* = \Hom(V,\bk)$.  From now on, whenever the linear dual $V^*$ is in sight, we tacitly assume that $\bk$ is a characteristic $0$ field.

Given a bilinear map $\mu \colon V^{\otimes 2} \to V$ and elements $x,y \in V$, we often write $\mu(x,y)$ as $xy$ and put in parentheses for longer products.  For a map $\Delta \colon V \to V^{\otimes 2}$, we use Sweedler's notation \cite{sweedler} for comultiplication: $\Delta(x) = \sum_{(x)} x_1 \otimes x_2$.  To simplify the typography in computations, we often omit the summation sign $\sum_{(x)}$.


\begin{definition}
\label{def:homas}
\begin{enumerate}
\item
A \textbf{Hom-associative algebra} \cite{ms} $(A,\mu,\alpha)$ consists of a  $\bk$-module $A$, a bilinear map $\mu \colon A^{\otimes 2} \to A$ (the multiplication), and a linear self-map $\alpha \colon A \to A$ such that (i) $\alpha \circ \mu = \mu \circ \alpha^{\otimes 2}$ (multiplicativity) and (ii) $\mu \circ (\alpha \otimes \mu) = \mu \circ (\mu \otimes \alpha)$ (Hom-associativity).
\item
A \textbf{Hom-coassociative coalgebra} \cite{ms2,ms4} $(C,\Delta,\alpha)$ consists of a $\bk$-module $C$, a linear map $\Delta \colon C \to C^{\otimes 2}$ (the comultiplication), and a linear self-map $\alpha \colon C \to C$ such that (i) $\alpha^{\otimes 2} \circ \Delta = \Delta \circ \alpha$ (comultiplicativity) and (ii) $(\alpha \otimes \Delta) \circ \Delta = (\Delta \otimes \alpha) \circ \Delta$ (Hom-coassociativity).
\item
A \textbf{Hom-bialgebra} \cite{ms2,yau3} is a quadruple $(A,\mu,\Delta,\alpha)$ in which $(A,\mu,\alpha)$ is a Hom-associative algebra, $(A,\Delta,\alpha)$ is a Hom-coassociative coalgebra, and the following condition holds:
\begin{equation}
\label{def:hombi}
\Delta \circ \mu = \mu^{\otimes 2} \circ (Id \otimes \tau \otimes Id) \circ \Delta^{\otimes 2}.
\end{equation}
\end{enumerate}
\end{definition}

The compatibility condition ~\eqref{def:hombi} can be restated as $\Delta(xy) = \sum_{(x)(y)} x_1y_1 \otimes x_2y_2$.

We will refer to $\alpha$ as the \emph{twisting map}.  In a Hom-bialgebra, the (co)multiplication is not (co)associative, but the non-(co)associativity is controlled by the twisting map $\alpha$.  In particular, a bialgebra is a Hom-bialgebra when equipped with $\alpha = Id$.  More generally, any bialgebra can be twisted into a Hom-bialgebra via a bialgebra morphism, as explained in the example below.

\begin{example}
\label{ex:homas}
\begin{enumerate}
\item
If $(A,\mu)$ is an associative algebra and $\alpha \colon A \to A$ is an algebra morphism, then $A_\alpha = (A,\mu_\alpha,\alpha)$ is a Hom-associative algebra with the twisted multiplication $\mu_\alpha = \alpha \circ \mu$ \cite{yau2}.  Indeed, the Hom-associativity axiom $\mu_\alpha \circ (\alpha \otimes \mu_\alpha) = \mu_\alpha \circ (\mu_\alpha \otimes \alpha)$ is equal to $\alpha^2$ applied to the associativity axiom of $\mu$.  Likewise, both sides of the multiplicativity axiom $\alpha \circ \mu_\alpha = \mu_\alpha \circ \alpha^{\otimes 2}$ are equal to $\alpha^2 \circ \mu$.
\item
Dually, if $(C,\Delta)$ is a coassociative coalgebra and $\alpha \colon C \to C$ is a coalgebra morphism, then $C_\alpha = (C,\Delta_\alpha,\alpha)$ is a Hom-coassociative coalgebra with the twisted comultiplication $\Delta_\alpha = \Delta \circ \alpha$.
\item
Combining the previous two cases, if $(A,\mu,\Delta)$ is a bialgebra and $\alpha \colon A \to A$ is a bialgebra morphism, then $A_\alpha = (A,\mu_\alpha,\Delta_\alpha,\alpha)$ is a Hom-bialgebra.  The compatibility condition \eqref{def:hombi} for $\Delta_\alpha = \Delta \circ \alpha$ and $\mu_\alpha = \alpha \circ \mu$ is straightforward to check.\qed
\end{enumerate}
\end{example}

It is clear that the axioms of a Hom-coassociative coalgebra are dual to those of a Hom-associative algebra.  The following examples have to do with this duality.

\begin{example}
\label{ex:duality}
\begin{enumerate}
\item
Let $(C,\Delta,\alpha)$ be a Hom-coassociative coalgebra and $C^*$ be the linear dual of $C$.  Then we have a Hom-associative algebra $(C^*,\Delta^*,\alpha^*)$, where
\begin{equation}
\label{deltadual}
\langle \Delta^*(\phi,\psi),x\rangle = \langle \phi \otimes \psi, \Delta(x)\rangle \quad \text{and} \quad
\alpha^*(\phi) = \phi \circ \alpha
\end{equation}
for all $\phi, \psi \in C^*$ and $x \in C$.  This is checked in exactly the same way as for (co)associative algebras \cite[2.1]{abe}, as was done in \cite[Corollary 4.12]{ms4}.
\item
Likewise, suppose that $(A,\mu,\alpha)$ is a finite dimensional Hom-associative algebra. Then $(A^*,\mu^*,\alpha^*)$ is a Hom-coassociative coalgebra, where
\begin{equation}
\label{mudual}
\langle \mu^*(\phi), x \otimes y\rangle = \langle \phi, \mu(x,y)\rangle  \quad \text{and} \quad
\alpha^*(\phi) = \phi \circ \alpha
\end{equation}
for all $\phi \in A^*$ and $x,y \in A$ \cite[Corollary 4.12]{ms4}.  In what follows, whenever $\mu^* \colon A^* \to A^* \otimes A^*$ is in sight, we tacitly assume that $A$ is finite dimensional.
\item
Combining the previous two examples, suppose that $(A,\mu,\Delta,\alpha)$ is a finite dimensional Hom-bialgebra.  Then so is $(A^*,\Delta^*,\mu^*,\alpha^*)$, where $\Delta^*$, $\mu^*$, and $\alpha^*$ are defined as in \eqref{deltadual} and \eqref{mudual}.\qed
\end{enumerate}
\end{example}

To generalize quantum groups to the Hom setting, we need a suitably weakened notion of a multiplicative identity for Hom-associative algebras.


\begin{definition}
\label{def:weakunit}
\begin{enumerate}
\item
Let $(A,\mu,\alpha)$ be a Hom-associative algebra.  A \textbf{weak unit} \cite{fg2} of $A$ is an element $c \in A$ such that $\alpha(x) = cx = xc$ for all $x \in A$.  In this case, we call $(A,\mu,\alpha,c)$ a \textbf{weakly unital Hom-associative algebra}.
\item
Let $(A,\mu,\alpha,c)$ be a weakly unital Hom-associative algebra and $R \in A^{\otimes 2}$.  Define the following elements in $A^{\otimes 3}$:
\begin{equation}
\label{R123}
R_{12} = R \otimes c,\quad R_{23} = c \otimes R, \quad R_{13} = (\tau \otimes Id)(R_{23}).
\end{equation}
\end{enumerate}
\end{definition}
The elements in \eqref{R123} are our Hom-type generalizations of the elements involved in the QYBE \eqref{qybe} and in the definition of a quasi-triangular bialgebra.

\begin{example}[\cite{fg2} Example 2.2]\label{ex:weakunit}
If $(A,\mu,1)$ is a unital associative algebra, then the Hom-associative algebra $A_\alpha = (A,\mu_\alpha,\alpha)$ (Example ~\ref{ex:homas}) has a weak unit $c = 1$.  In fact, we have
\[
\mu_\alpha(1,x) = \alpha(\mu(1,x)) = \alpha(x) = \alpha(\mu(x,1)) = \mu_\alpha(x,1).
\]
So $(A,\mu_\alpha,\alpha,1)$ is a weakly unital Hom-associative algebra.\qed
\end{example}


We are now ready for the main definition of this paper, which gives a non-associative and non-coassociative version of a quasi-triangular bialgebra.

\begin{definition}
\label{def:qthb}
A \textbf{quasi-triangular Hom-bialgebra} is a tuple $(A,\mu,\Delta,\alpha,c,R)$ in which $(A,\mu,\Delta,\alpha)$ is a Hom-bialgebra, $c$ is a weak unit of $(A,\mu,\alpha)$, and $R \in A^{\otimes 2}$ satisfies the following three axioms:
\begin{subequations}
\label{qtaxioms}
\begin{align}
(\Delta \otimes \alpha)(R) &= R_{13}R_{23},\label{R1323}\\
(\alpha \otimes \Delta)(R) &= R_{13}R_{12},\label{R1312}\\
[(\tau\circ\Delta)(x)]R &= R\Delta(x)\label{RDelta}
\end{align}
\end{subequations}
for all $x \in A$.  The elements $R_{12}$, $R_{13}$, and $R_{23}$ are defined in \eqref{R123}.
\end{definition}

\begin{example}
A quasi-triangular bialgebra $(A,R)$ \eqref{qtb} in the sense of Drinfel'd \cite{dri87} becomes a quasi-triangular Hom-bialgebra when equipped with $\alpha = Id$ and $c = 1$.  In a quasi-triangular bialgebra, it is usually assumed that the quasi-triangular structure $R$ is invertible.  In that case, the axiom \eqref{RDelta} can be restated as $(\tau \circ \Delta)(x) = R\Delta(x)R^{-1}$.  However, in this paper, even when we refer to a quasi-triangular bialgebra (e.g., in Theorem ~\ref{thm:twist1} and Corollary ~\ref{cor:twist}), we will not have to use its counit or the invertibility of its quasi-triangular structure.  See also Remark ~\ref{rk:R}.\qed
\end{example}

Some remarks are in order.

\begin{remark}
\begin{enumerate}
\item
The multiplications in \eqref{qtaxioms} are computed in each tensor factor using the multiplication $\mu$ in $A$.  They all make sense because there is no iterated multiplication.
\item
Since a weak unit $c$ (Definition ~\ref{def:weakunit}) is not actually a multiplicative identity, it does not make much sense to talk about multiplicative inverse in a weakly unital Hom-associative algebra.  This is the reason for not insisting on $R$ (in Definition ~\ref{def:qthb}) being invertible and for stating \eqref{RDelta} without using $R^{-1}$.
\item
Write $R = \sum s_i \otimes t_i \in A^{\otimes 2}$.  From the definitions \eqref{R123} of the $R_{ij}$ and the requirement that $c$ be a weak unit, the three axioms in \eqref{qtaxioms} can be restated as:
\begin{subequations}
\label{qtax}
\begin{align}
(\Delta \otimes \alpha)(R) &= \sum \alpha(s_i) \otimes \alpha(s_j) \otimes t_it_j,\label{R1323'}\\
(\alpha \otimes \Delta)(R) &= \sum s_js_i \otimes \alpha(t_i) \otimes \alpha(t_j),\label{R1312'}\\
\sum x_2s_i \otimes x_1t_i &= \sum s_ix_1 \otimes t_ix_2,
\end{align}
\end{subequations}
where $\Delta(x) = \sum x_1 \otimes x_2$.
\end{enumerate}
\end{remark}


A major reason for introducing quasi-triangular bialgebras $(A,R)$ \cite{dri87} is that the quasi-triangular structure $R$ satisfies the QYBE \eqref{qybe}.  As we mentioned in the introduction, the fact that $R$ is a solution of the QYBE leads to solutions of the YBE \eqref{ybe} in the representations of $A$.  We will generalize this relationship between the QYBE and the YBE to the Hom setting in section ~\ref{sec:hybe}.  As a first step, we now show that the element $R$ in a quasi-triangular Hom-bialgebra satisfies two non-associative versions of the QYBE.

\begin{theorem}
\label{thm:qhybe}
Let $(A,\mu,\Delta,\alpha,c,R)$ be a quasi-triangular Hom-bialgebra.  Then $R$ satisfies the quantum Hom-Yang-Baxter equations (QHYBE)
\begin{equation}
\label{qhybe}
(R_{12}R_{13})R_{23} = R_{23}(R_{13}R_{12})
\end{equation}
and
\begin{equation}
\label{qhybe'}
R_{12}(R_{13}R_{23}) = (R_{23}R_{13})R_{12}.
\end{equation}
\end{theorem}

\begin{proof}
First consider \eqref{qhybe}.  Recall that $R_{23} = c \otimes R$, and we have
\begin{equation}
\label{1213}
R_{12}R_{13} = (Id \otimes \tau)(R_{13}R_{12}).
\end{equation}
We compute as follows:
\begin{subequations}
\allowdisplaybreaks
\begin{align*}
(R_{12}R_{13})R_{23}
&= [(Id \otimes \tau)(R_{13}R_{12})]R_{23} \quad \text{by \eqref{1213}}\\
&= [(Id \otimes \tau)(\alpha \otimes \Delta)(R)]R_{23} \quad \text{by \eqref{R1312}}\\
&= [(\alpha \otimes (\tau\circ\Delta))(R)](c \otimes R)\\
&= (c \otimes R)[(\alpha \otimes \Delta)(R)] \quad \text{by \eqref{RDelta}}\\
&= R_{23}(R_{13}R_{12}) \quad \text{by \eqref{R1312}}.
\end{align*}
\end{subequations}
This proves that $R$ satisfies the QHYBE \eqref{qhybe}.

The other QHYBE \eqref{qhybe'} is proved by a similar computation.  Since $R_{12} = R \otimes c$ and
\begin{equation}
\label{1323}
(\tau \otimes Id)(R_{13}R_{23}) = R_{23}R_{13},
\end{equation}
we have:
\begin{subequations}
\allowdisplaybreaks
\begin{align*}
R_{12}(R_{13}R_{23})
&= (R \otimes c)[(\Delta \otimes \alpha)(R)] \quad\text{by \eqref{R1323}}\\
&= [((\tau\circ\Delta)\otimes\alpha)(R)](R \otimes c) \quad\text{by \eqref{RDelta}}\\
&= [(\tau \otimes Id) \circ (\Delta \otimes \alpha)(R)]R_{12}\\
&= [(\tau \otimes Id)(R_{13}R_{23})]R_{12} \quad\text{by \eqref{R1323}}\\
&= (R_{23}R_{13})R_{12}  \quad \text{by \eqref{1323}}.
\end{align*}
\end{subequations}
This finishes the proof.
\end{proof}

\begin{remark}
Let us make it clear that the two quantum Hom-Yang-Baxter equations \eqref{qhybe} and \eqref{qhybe'} are indeed different in general.  Writing $R = \sum s_i \otimes t_i$, the left-hand side in \eqref{qhybe} is
\begin{equation}
\label{qlhs}
\begin{split}
(R_{12}R_{13})R_{23}
&= (s_js_i \otimes t_jc \otimes ct_i)(c \otimes s_k \otimes t_k)\\
&= (s_js_i \otimes \alpha(t_j) \otimes \alpha(t_i))(c \otimes s_k \otimes t_k)\\
&= \alpha(s_js_i) \otimes \alpha(t_j)s_k \otimes \alpha(t_i)t_k\\
&= \alpha(s_j)\alpha(s_i) \otimes \alpha(t_j)s_k \otimes \alpha(t_i)t_k.
\end{split}
\end{equation}
Likewise, the left-hand side in \eqref{qhybe'} is
\begin{equation}
\label{qlhs'}
\begin{split}
R_{12}(R_{13}R_{23})
&= (s_j \otimes t_j \otimes c)(s_ic \otimes cs_k \otimes t_it_k)\\
&= (s_j \otimes t_j \otimes c)(\alpha(s_i) \otimes \alpha(s_k) \otimes t_it_k)\\
&= s_j\alpha(s_i) \otimes t_j \alpha(s_k) \otimes \alpha(t_it_k)\\
&= s_j\alpha(s_i) \otimes t_j \alpha(s_k) \otimes \alpha(t_i)\alpha(t_k).
\end{split}
\end{equation}
One observes that the last lines in \eqref{qlhs} and in \eqref{qlhs'} are different.  A similar computation for $R_{23}(R_{13}R_{12})$ and $(R_{23}R_{13})R_{12}$ shows that they are different as well.
\end{remark}

\begin{remark}
\label{rk:alphainv}
On the other hand, suppose that $R$ is $\alpha$-invariant, i.e., $\alpha^{\otimes 2}(R) = \sum \alpha(s_i) \otimes \alpha(t_i) = R$.  Then one can see from \eqref{qlhs} and \eqref{qlhs'} (and a similar computation for $R_{23}(R_{13}R_{12})$ and $(R_{23}R_{13})R_{12}$) that the two quantum Hom-Yang-Baxter equations \eqref{qhybe} and \eqref{qhybe'} coincide.
\end{remark}


We finish this section with some alternative characterizations of the axioms \eqref{R1323} and \eqref{R1312} of a quasi-triangular Hom-bialgebra.  They generalize an observation in \cite[p.812 (5)]{dri87}.  Let $(A,\mu,\Delta,\alpha,c)$ be a Hom-bialgebra with a weak unit $c$, $R \in A^{\otimes 2}$ be an arbitrary element, and $A^*$ be the linear dual of $A$.  Define four linear maps $\lambda_1, \lambda_1',\lambda_2,\lambda_2' \colon A^* \to A$ by setting
\begin{equation}
\label{lambdamaps}
\begin{split}
\lambda_1(\phi) &= \langle\phi \otimes \alpha,R\rangle, \quad
\lambda_1'(\phi) = \langle \alpha^*(\phi) \otimes Id,R\rangle,\\
\lambda_2(\phi) &= \langle \alpha \otimes \phi,R\rangle, \quad
\lambda_2'(\phi) = \langle Id \otimes \alpha^*(\phi),R\rangle
\end{split}
\end{equation}
for $\phi \in A^*$.

In the following characterizations of the axioms \eqref{R1323} and \eqref{R1312}, we use the operations $\Delta^* \colon A^* \otimes A^* \to A^*$ \eqref{deltadual} and $\mu^* \colon A^* \to A^* \otimes A^*$ \eqref{mudual} (when $A$ is finite dimensional) discussed in Example ~\ref{ex:duality} above.

\begin{theorem}
\label{thm:axiom1}
Let $(A,\mu,\Delta,\alpha,c)$ be a Hom-bialgebra with a weak unit $c$ and $R \in A^{\otimes 2}$ be an arbitrary element.  With the notations \eqref{lambdamaps}, the following statements are equivalent.
\begin{enumerate}
\item
The axiom \eqref{R1323} holds, i.e., $(\Delta \otimes \alpha)(R) = R_{13}R_{23}$.
\item
The diagram
\begin{equation}
\label{axiom1'}
\SelectTips{cm}{10}
\xymatrix{
A^* \otimes A^* \ar[rr]^-{\lambda_1' \otimes \lambda_1'} \ar[d]_-{\Delta^*} & & A \otimes A \ar[d]^-{\mu}\\
A^* \ar[rr]^-{\lambda_1} & & A
}
\end{equation}
is commutative.
\end{enumerate}
If $A$ is finite dimensional, then the two statements above are also equivalent to the commutativity of the diagram
\begin{equation}
\label{axiom1''}
\SelectTips{cm}{10}
\xymatrix{
A^* \ar[rr]^-{\lambda_2'} \ar[d]_-{\mu^*} & & A \ar[d]^-{\Delta}\\
A^* \otimes A^* \ar[rr]^-{\lambda_2 \otimes \lambda_2} & & A \otimes A.
}
\end{equation}
\end{theorem}

\begin{proof}
First we show the equivalence between the axiom ~\eqref{R1323} and the commutativity of the square \eqref{axiom1'}. Write $R = \sum s_i \otimes t_i$.  Axiom \eqref{R1323} holds if and only if
\begin{equation}
\label{ax1}
\langle \phi \otimes \psi \otimes Id, (\Delta \otimes \alpha)(R)\rangle =
\langle \phi \otimes \psi \otimes Id, R_{13}R_{23}\rangle
\end{equation}
for all $\phi, \psi \in A^*$.  The left-hand side in \eqref{ax1} is:
\begin{subequations}
\allowdisplaybreaks
\begin{align*}
\langle \phi \otimes \psi \otimes Id, (\Delta \otimes \alpha)(R)\rangle
&= \langle \phi \otimes \psi, \Delta(s_i)\rangle\alpha(t_i)\\
&= \langle \Delta^*(\phi,\psi),s_i\rangle\alpha(t_i)\\
&= \langle \Delta^*(\phi,\psi) \otimes \alpha, s_i \otimes t_i\rangle\\
&= \lambda_1(\Delta^*(\phi,\psi)).
\end{align*}
\end{subequations}
On the other hand, we have $R_{13}R_{23} = \sum \alpha(s_i) \otimes \alpha(s_j) \otimes t_it_j$ \eqref{R1323'}.  So the right-hand side in \eqref{ax1} is:
\begin{subequations}
\allowdisplaybreaks
\begin{align*}
\langle \phi \otimes \psi \otimes Id, R_{13}R_{23}\rangle
&= \langle \phi \otimes \psi \otimes Id,  \alpha(s_i) \otimes \alpha(s_j) \otimes t_it_j\rangle\\
&= \langle \phi, \alpha(s_i)\rangle\langle \psi,\alpha(s_j)\rangle t_it_j\\
&= \left(\langle \alpha^*(\phi),s_i\rangle t_i\right)\left(\langle \alpha^*(\psi),s_j\rangle t_j\right)\\
&= \left(\langle \alpha^*(\phi) \otimes Id, s_i \otimes t_i \rangle\right)\left(\langle \alpha^*(\psi) \otimes Id, s_j \otimes t_j\rangle\right)\\
&= \mu(\lambda_1'(\phi),\lambda_1'(\psi)).
\end{align*}
\end{subequations}
Therefore, the condition \eqref{ax1} holds for all $\phi,\psi \in A^*$ if and only if the square \eqref{axiom1'} is commutative.

Next we show the equivalence between the axiom ~\eqref{R1323} and the commutativity of the square \eqref{axiom1''} when $A$ is finite dimensional.  The finite dimensionality of $A$ ensures that $\mu^*$ is well-defined.  Axiom \eqref{R1323} holds if and only if
\begin{equation}
\label{ax1'}
\langle Id \otimes Id \otimes \phi, (\Delta \otimes \alpha)(R)\rangle =
\langle Id \otimes Id \otimes \phi, R_{13}R_{23}\rangle
\end{equation}
for all $\phi \in A^*$.  The left-hand side in \eqref{ax1'} is:
\begin{subequations}
\allowdisplaybreaks
\begin{align*}
\langle Id \otimes Id \otimes \phi, (\Delta \otimes \alpha)(R)\rangle &= \Delta(s_i)\langle \phi, \alpha(t_i)\rangle\\
&= \Delta(s_i\langle \alpha^*(\phi),t_i\rangle)\\
&= \Delta(\langle Id \otimes \alpha^*(\phi),s_i \otimes t_i\rangle\\
&= \Delta(\lambda_2'(\phi)).
\end{align*}
\end{subequations}
Below we write $\mu^*(\phi) = \sum \phi_1 \otimes \phi_2$.  The right-hand side in \eqref{ax1'} is:
\begin{subequations}
\allowdisplaybreaks
\begin{align*}
\langle Id \otimes Id \otimes \phi, R_{13}R_{23}\rangle
&= \langle Id \otimes Id \otimes \phi, \alpha(s_i) \otimes \alpha(s_j) \otimes t_it_j\rangle\\
&= \alpha(s_i) \otimes \alpha(s_j) \langle \mu^*(\phi), t_i \otimes t_j\rangle\\
&= \alpha(s_i)\langle \phi_1,t_i\rangle \otimes \alpha(s_j)\langle \phi_2, t_j\rangle\\
&= \langle \alpha \otimes \phi_1, s_i \otimes t_i\rangle \otimes \langle \alpha \otimes \phi_2, s_j \otimes t_j\rangle\\
&= \lambda_2^{\otimes 2}(\mu^*(\phi)).
\end{align*}
\end{subequations}
Therefore, the condition \eqref{ax1'} holds for all $\phi \in A^*$ if and only if the square \eqref{axiom1''} is commutative.
\end{proof}

The following result is the analogue of Theorem ~\ref{thm:axiom1} for the axiom \eqref{R1312}.

\begin{theorem}
\label{thm:axiom2}
Let $(A,\mu,\Delta,\alpha,c)$ be a Hom-bialgebra with a weak unit $c$ and $R \in A^{\otimes 2}$ be an arbitrary element.  With the notations \eqref{lambdamaps}, the following statements are equivalent.
\begin{enumerate}
\item
The axiom \eqref{R1312} holds, i.e., $(\alpha \otimes \Delta)(R) = R_{13}R_{12}$.
\item
The diagram
\begin{equation}
\label{axiom2'}
\SelectTips{cm}{10}
\xymatrix{
A^* \otimes A^* \ar[rr]^-{\lambda_2' \otimes \lambda_2'} \ar[d]_-{\Delta^*} & & A \otimes A \ar[d]^-{\mu^{op}}\\
A^* \ar[rr]^-{\lambda_2} & & A
}
\end{equation}
is commutative, where $\mu^{op} = \mu \circ \tau$.
\end{enumerate}
If $A$ is finite dimensional, then the two statements above are also equivalent to the commutativity of the diagram
\begin{equation}
\label{axiom2''}
\SelectTips{cm}{10}
\xymatrix{
A^* \ar[rr]^-{\lambda_1'} \ar[d]_-{\mu^{*op}} & & A \ar[d]^-{\Delta}\\
A^* \otimes A^* \ar[rr]^-{\lambda_1 \otimes \lambda_1} & & A \otimes A,
}
\end{equation}
where $\mu^{*op} = \tau \circ \mu^*$.
\end{theorem}

\begin{proof}
This proof is completely analogous to that of Theorem ~\ref{thm:axiom1}, so we only give a sketch. The axiom \eqref{R1312} is equivalent to the equality
\[
\langle Id \otimes \phi \otimes \psi, (\alpha \otimes \Delta)(R)\rangle = \langle Id \otimes \phi \otimes \psi, R_{13}R_{12}\rangle
\]
for all $\phi, \psi \in A^*$.  Some calculation shows that this equality is in turn equivalent to the commutativity of the square \eqref{axiom2'}.  To show the equivalence between \eqref{R1312} and the commutativity of the square \eqref{axiom2''}, one uses $\phi \otimes Id \otimes Id$ instead of $Id \otimes \phi \otimes \psi$.
\end{proof}

In the special case $\alpha = Id$, we have $\lambda_1 = \lambda_1'$ and $\lambda_2 = \lambda_2'$.  So in this case, the commutative diagrams \eqref{axiom1'}, \eqref{axiom1''}, \eqref{axiom2'}, and \eqref{axiom2''} mean, respectively, that $\lambda_1$ is an algebra morphism, that $\lambda_2$ is a coalgebra morphism, that $\lambda_2$ is an algebra anti-morphism, and that $\lambda_1$ is a coalgebra anti-morphism.

\section{Examples of quasi-triangular Hom-bialgebras}
\label{sec:example}

Before we discuss the relationships between the QHYBE (\eqref{qhybe} and \eqref{qhybe'}) and the Hom version of the YBE, in this section we describe several classes of quasi-triangular Hom-bialgebras (Examples ~\ref{ex:anyon} - ~\ref{ex:sl2}).


We begin with some general twisting procedures by which quasi-triangular Hom-bialgebras can be constructed (Theorem ~\ref{thm:twist1} - Corollary ~\ref{cor:twist}).  The first twisting procedure, applied to the (co)multiplication, produces a family of quasi-triangular Hom-bialgebras from any given quasi-triangular bialgebra.  Recall the definitions of a quasi-triangular Hom-bialgebra (Definition ~\ref{def:qthb}) and of a quasi-triangular bialgebra (in the paragraph containing \eqref{qybe}).

\begin{theorem}
\label{thm:twist1}
Let $(A,\mu,\Delta,R)$ be a quasi-triangular bialgebra and $\alpha \colon A \to A$ be a bialgebra morphism (not-necessarily preserving $1$).  Then
\[
A_\alpha = (A,\mu_\alpha,\Delta_\alpha,\alpha,1,R)
\]
is a quasi-triangular Hom-bialgebra, in which $\mu_\alpha = \alpha \circ \mu$ and $\Delta_\alpha = \Delta \circ \alpha$.
\end{theorem}

\begin{proof}
As noted in Examples ~\ref{ex:homas} and ~\ref{ex:weakunit}, $(A,\mu_\alpha,\Delta_\alpha,\alpha,1)$ is a Hom-bialgebra with a weak unit $c = 1$.  It remains to verify the three axioms \eqref{qtaxioms} for $A_\alpha$.

For \eqref{R1323} note that, since $R = \sum s_i \otimes t_i$ is a quasi-triangular structure \eqref{qtb}, we have
\[
(\Delta \otimes Id)(R) = R_{13}R_{23} = \sum s_i \otimes s_j \otimes t_it_j,
\]
where $t_it_j = \mu(t_i,t_j)$.  Also, we have $\Delta_\alpha = \alpha^{\otimes 2} \circ \Delta$ because $\alpha$ is a bialgebra morphism.  Therefore, we have
\begin{subequations}
\allowdisplaybreaks
\begin{align*}
(\Delta_\alpha \otimes \alpha)(R)
&= ((\alpha^{\otimes 2} \circ \Delta) \otimes \alpha)(R)\\
&= \alpha^{\otimes 3}(\Delta \otimes Id)(R)\\
&= \alpha^{\otimes 3} \left(s_i \otimes s_j \otimes t_it_j\right)\\
&= \alpha(s_i) \otimes \alpha(s_j) \otimes \mu_\alpha(t_i,t_j).
\end{align*}
\end{subequations}
This proves \eqref{R1323} (in the alternative formulation \eqref{R1323'}) for $A_\alpha$.  Similarly, for \eqref{R1312} we use
\[
(Id \otimes \Delta)(R) = R_{13}R_{12} = \sum s_js_i \otimes t_i \otimes t_j.
\]
This gives
\begin{subequations}
\allowdisplaybreaks
\begin{align*}
(\alpha \otimes \Delta_\alpha)(R)
&= (\alpha \otimes (\alpha^{\otimes 2} \circ \Delta))(R)\\
&= \alpha^{\otimes 3}\left(Id \otimes \Delta\right)(R)\\
&= \alpha^{\otimes 3}\left(s_js_i \otimes t_i \otimes t_j\right)\\
&= \mu_\alpha(s_j,s_i) \otimes \alpha(t_i) \otimes \alpha(t_j).
\end{align*}
\end{subequations}
This proves \eqref{R1312} (in the alternative formulation \eqref{R1312'}) for $A_\alpha$.

For \eqref{RDelta} note that we have $((\tau \circ \Delta)(y))R = R\Delta(y)$ \eqref{qtb} for $y \in A$, i.e.,
\[
\sum y_2s_i \otimes y_1t_i = \sum s_iy_1 \otimes t_iy_2,
\]
where $\Delta(y) = \sum y_1 \otimes y_2$.  We use it below when $y = \alpha(x)$ for $x \in A$.  We have
\begin{subequations}
\allowdisplaybreaks
\begin{align*}
\mu_\alpha((\tau \circ \Delta_\alpha)(x),R)
&= \mu_\alpha(\alpha(x)_2 \otimes \alpha(x)_1, s_i \otimes t_i)\\
&= \alpha\left(\alpha(x)_2s_i\right) \otimes \alpha\left(\alpha(x)_1t_i\right)\\
&= \alpha\left(s_i\alpha(x)_1\right) \otimes \alpha\left(t_i\alpha(x)_2\right)\\
&= \mu_\alpha(R, \Delta_\alpha(x)).
\end{align*}
\end{subequations}
This proves \eqref{RDelta} for $A_\alpha$.
\end{proof}

\begin{remark}
\label{rk:R}
In the proof of Theorem ~\ref{thm:twist1}, we did not use the invertibility of $R$.  So Theorem ~\ref{thm:twist1} is still true even if $R$ is not invertible.  The same goes for Corollary ~\ref{cor:twist} below.
\end{remark}


The second twisting procedure, applied to the element $R$, produces a family of quasi-triangular Hom-bialgebras from any given quasi-triangular Hom-bialgebra with a surjective twisting map.

\begin{theorem}
\label{thm:twist2}
Let $(A,\mu,\Delta,\alpha,c,R)$ be a quasi-triangular Hom-bialgebra with $\alpha$ surjective and $n$ be a positive integer.  Then
\[
\an = (A,\mu,\Delta,\alpha,c,\rn)
\]
is also a quasi-triangular Hom-bialgebra, where $\rn = (\alpha^n \otimes \alpha^n)(R)$.
\end{theorem}

\begin{proof}
By induction it suffices to prove the case $n=1$.  We need to check the three axioms ~\eqref{qtaxioms} for $\ralpha = (\alpha \otimes \alpha)(R) = \sum \alpha(s_i) \otimes \alpha(t_i)$, where $R = \sum s_i \otimes t_i$.  Using the assumption that $\alpha$ is (co)multiplicative and that $c$ is a weak unit, we compute as follows:
\begin{subequations}
\allowdisplaybreaks
\begin{align*}
(\Delta \otimes \alpha)(\ralpha)
&= (\Delta \otimes \alpha)(\alpha \otimes \alpha)(R)\\
&= \alpha^{\otimes 3}((\Delta \otimes \alpha)(R))\\
&= \alpha^{\otimes 3}(\alpha(s_i) \otimes \alpha(s_j) \otimes t_it_j) \quad\text{by \eqref{R1323'}}\\
&= \alpha^2(s_i) \otimes \alpha^2(s_j) \otimes \alpha(t_it_j)\\
&= \alpha(s_i)c \otimes c \alpha(s_j) \otimes \alpha(t_i)\alpha(t_j)\\
&= (\alpha(s_i) \otimes c \otimes \alpha(t_i))(c \otimes \alpha(s_j) \otimes \alpha(t_j))\\
&= \ralpha_{13}\ralpha_{23}.
\end{align*}
\end{subequations}
This proves \eqref{R1323} for $\ralpha$.  The axiom \eqref{R1312} for $\ralpha$ is proved similarly.  Notice that we have not used the surjectivity assumption of $\alpha$ so far.

To prove \eqref{RDelta} for $\ralpha$, pick an element $x \in A$.  Since $\alpha$ is assumed to be surjective, we have $x = \alpha(y)$ for some (not-necessarily unique) element $y \in A$.  By the comultiplicativity of $\alpha$, we have
\[
\Delta(x) = \sum x_1 \otimes x_2 = \sum \alpha(y)_1 \otimes \alpha(y)_2 = \sum \alpha(y_1) \otimes \alpha(y_2) = \alpha^{\otimes 2}(\Delta(y)).
\]
Using the multiplicativity of $\alpha$, we compute as follows:
\begin{subequations}
\allowdisplaybreaks
\begin{align*}
[(\tau \circ \Delta)(x)]\ralpha
&= \alpha(y_2)\alpha(s_i) \otimes \alpha(y_1)\alpha(t_i)\\
&= \alpha^{\otimes 2}(((\tau\circ\Delta)(y))R)\\
&= \alpha^{\otimes 2}(R\Delta(y))\quad\text{by \eqref{RDelta}}\\
&= \alpha(s_i)\alpha(y_1) \otimes \alpha(t_i)\alpha(y_2)\\
&= \alpha(s_i)x_1 \otimes \alpha(t_i)x_2\\
&= \ralpha\Delta(x).
\end{align*}
\end{subequations}
This proves \eqref{RDelta} for $\ralpha$.
\end{proof}

The following result is an immediate consequence of Theorems ~\ref{thm:twist1} and ~\ref{thm:twist2}.

\begin{corollary}
\label{cor:twist}
Let $(A,\mu,\Delta,R)$ be a quasi-triangular bialgebra, $\alpha \colon A \to A$ be a surjective bialgebra morphism (not-necessarily preserving $1$), and $n$ be a positive integer.  Then
\[
A_\alpha^{(n)} = (A,\mu_\alpha,\Delta_\alpha,\alpha,1,\rn)
\]
is a quasi-triangular Hom-bialgebra, where $\mu_\alpha = \alpha \circ \mu$, $\Delta_\alpha = \Delta \circ \alpha$, and $\rn = (\alpha^n \otimes \alpha^n)(R)$.
\end{corollary}

We now give a series of examples of quasi-triangular Hom-bialgebras using Theorem ~\ref{thm:twist1} and Corollary ~\ref{cor:twist}.

\begin{example}[\textbf{Anyon-generating Hom-quantum groups}]
\label{ex:anyon}
Consider the group bialgebra $\bC \bZ/n$ over $\bC$ generated by the cyclic group $\bZ/n$ with generator $g$ and relation $g^n = 1$.  Its comultiplication is determined by
\begin{equation}
\label{kgdelta}
\Delta(g) = g \otimes g,
\end{equation}
which is cocommutative.  It becomes a quasi-triangular bialgebra when equipped with the non-trivial quasi-triangular structure
\[
R = \frac{1}{n} \sum_{p,q=0}^{n-1} \exp(-2\pi ipq/n) g^p \otimes g^q.
\]
The quasi-triangular bialgebra $(\bC\bZ/n,R)$ is called the \emph{anyon-generating quantum group} (\cite{majid92} and \cite[Example 2.1.6]{majid}) and is denoted by $\zn$.  We can get bialgebra morphisms on $\bC\bZ/n$ by extending the group morphisms
\[
\alpha_k \colon \bZ/n \to \bZ/n, \quad \alpha_k(g) = g^k
\]
for $k \in \left\{1, \ldots , n-1\right\}$.  Moreover, $\alpha_k$ is surjective if and only if $k$ and $n$ are relatively prime.

By Theorem ~\ref{thm:twist1}, for each $k \in \left\{1, \ldots , n-1\right\}$ we have a quasi-triangular Hom-bialgebra
\[
(\zn)_{\alpha_k} = (\bC\bZ/n,\mu_{\alpha_k},\Delta_{\alpha_k},\alpha_k,1,R)
\]
with twisted (co)multiplication.  Here $\mu_{\alpha_k} = \alpha_k\circ\mu$ (with $\mu$ the multiplication in $\bC\bZ/n$) and $\Delta_{\alpha_k}$ is determined by $\Delta_{\alpha_k}(g) = g^k \otimes g^k$.

Suppose that $k$ and $n$ are relatively prime, so $\alpha_k$ is surjective.  By Corollary ~\ref{cor:twist}, for each integer $t \geq 1$ we have a quasi-triangular Hom-bialgebra
\[
(\zn)_{\alpha_k}^{(t)} = (\bC\bZ/n,\mu_{\alpha_k},\Delta_{\alpha_k},\alpha_k,1,R^{\alpha^t}),
\]
where the twisted quasi-triangular structure is
\[
R^{\alpha^t} = (\alpha_k^t \otimes \alpha_k^t)(R) = \frac{1}{n} \sum_{p,q=0}^{n-1} \exp(-2\pi ipq/n) g^{pk^t} \otimes g^{qk^t}.
\]
\qed
\end{example}

\begin{example}[\textbf{Hom-quantum group bialgebras}]
\label{ex:groupbialg}
This is a generalization of the previous example.  Let $\bk$ be a characteristic $0$ field,  $G$ be a finite abelian group, written multiplicatively with identity $e$, and $\kg$ be its group bialgebra \cite[p.58, Example 2.4]{abe}.  Its comultiplication is determined by \eqref{kgdelta} for $g \in G$.  Since $G$ is commutative, $\kg$ is both commutative and cocommutative.

Identify $\kg \otimes \kg$ with $\bk H$, where $H = G \times G$.  A (not-necessarily invertible) quasi-triangular structure on the group bialgebra $\kg$ is equivalent to a function $R \colon G \times G \to \bk$ such that
\begin{equation}
\label{eq:groupbi}
\sum_{xy=v} R(u,x)R(w,y) = \delta_{u,w}R(u,v) \quad\text{and}\quad
\sum_{xy=u} R(x,v)R(y,w) = \delta_{v,w}R(u,v)
\end{equation}
for all $u,v,w \in G$ \cite[Example 2.1.17]{majid}, where $\delta_{u,w}$ denotes the Kronecker delta.  In fact, writing $R = \sum_{u,v\in G} R(u,v) u \otimes v$, the two conditions in \eqref{eq:groupbi} are equivalent to the axioms $(\Delta \otimes Id)(R) = R_{13}R_{23}$ and $(Id \otimes \Delta)(R) = R_{13}R_{12}$ \eqref{qtb}, respectively.  The condition $(\tau \circ \Delta(x))(R) = R\Delta(x)$ is automatic because $\kg$ is both commutative and cocommutative.  Fix such a quasi-triangular structure $R$ for the rest of this example.

If $\alpha \colon G \to G$ is any group morphism, then it extends naturally to a bialgebra morphism on $\kg$, where $\alpha\left(\sum_u c_uu\right) = \sum_u c_u\alpha(u)$.  By Theorem ~\ref{thm:twist1} (and Remark ~\ref{rk:R}) we have a quasi-triangular Hom-bialgebra
\[
\kg_\alpha = (\kg,\mu_\alpha = \alpha \circ \mu,\Delta_\alpha = \Delta \circ \alpha,\alpha,e,R),
\]
where $\mu$ and $\Delta$ are the multiplication and the comultiplication in $\kg$, respectively.

Suppose, in addition, that $\alpha \colon G \to G$ is a group automorphism.  Then by Corollary ~\ref{cor:twist} (and Remark ~\ref{rk:R}) we have a quasi-triangular Hom-bialgebra
\[
\kg_\alpha^{(n)} = (\kg,\mu_\alpha,\Delta_\alpha,\alpha,e,\rn)
\]
for each $n \geq 1$.   We can make the twisted quasi-triangular structure $\rn = (\alpha^n \otimes \alpha^n)(R)$ more explicit as follows.  Thinking of $R$ as $\sum_{u,v} R(u,v) u \otimes v$, we have
\[
\alpha^{\otimes 2}(R) = \sum_{u,v}R(u,v) \alpha(u) \otimes \alpha(v) = \sum_{u,v} R(\alpha^{-1}(u),\alpha^{-1}(v)) u \otimes v,
\]
since $\alpha$ is invertible.  Therefore, $\rn$ is equivalent to the function $G \times G \to \bk$ given by
\begin{equation}
\label{rnfunction}
\rn(u,v) = R\left((\alpha^{-1})^n(u), (\alpha^{-1})^n(v)\right)
\end{equation}
for $u,v \in G$.
\qed
\end{example}

\begin{example}[\textbf{Hom-quantum function bialgebras}]
\label{ex:function}
This example is closely related to the previous example.  Let $G$ be a finite abelian group, $\bk$ be a characteristic $0$ field, and $\kofg$ be the bialgebra of functions $G \to \bk$ \cite[p.58, Example 2.4]{abe}.  Its multiplication $\mu$ is defined point-wise, i.e.,
\[
\langle \mu(\phi,\psi), u\rangle = \langle \phi,u\rangle \langle \psi,u\rangle,
\]
for $u \in G$ and $\phi, \psi \in \kofg$.  Its comultiplication $\Delta$ is dual to the multiplication in $G$, i.e.,
\[
\langle \Delta(\phi), (u,v)\rangle = \langle \phi, uv\rangle
\]
for $\phi \in \kofg$ and $u,v \in G$.  Since $G$ is commutative, $\kofg$ is both commutative and cocommutative.

A (not-necessarily invertible) quasi-triangular structure $R$ on the function bialgebra $\kofg$ is equivalent to a function $R \colon G \times G \to \bk$ such that
\begin{equation}
\label{eq:functionbi}
R(uv,w) = R(u,w)R(v,w) \quad\text{and}\quad
R(u,vw) = R(u,w)R(u,v)
\end{equation}
for all $u,v,w \in G$ \cite[Example 2.1.18]{majid}.  In fact, the two conditions in \eqref{eq:functionbi} are equivalent to the axioms $(\Delta \otimes Id)(R) = R_{13}R_{23}$ and $(Id \otimes \Delta)(R) = R_{13}R_{12}$ \eqref{qtb}, respectively.  The condition $(\tau \circ \Delta(x))(R) = R\Delta(x)$ is automatic because $\kofg$ is both commutative and cocommutative.  Fix such a quasi-triangular structure $R$ for the rest of this example.

If $\alpha \colon G \to G$ is a group morphism, then it extends naturally to a bialgebra morphism $\alpha^* \colon \kofg \to \kofg$ given by $\alpha^*(\phi) = \phi \circ \alpha$.  By Theorem ~\ref{thm:twist1} (and Remark ~\ref{rk:R}) we have a quasi-triangular Hom-bialgebra
\[
\kofg_{\alpha^*} = (\kofg,\mu_{\alpha^*} = \alpha^* \circ \mu,\Delta_{\alpha^*} = \Delta \circ \alpha^*,\alpha^*,e,R).
\]
The twisted multiplication and comultiplication are given by
\[
\langle \mu_{\alpha^*}(\phi,\psi), u\rangle = \langle \phi,\alpha(u)\rangle \langle \psi,\alpha(u)\rangle \quad\text{and}\quad
\langle \Delta_{\alpha^*}(\phi), (u,v)\rangle = \langle \phi, \alpha(uv)\rangle.
\]

If the group morphism $\alpha \colon G \to G$ is invertible, then so is the induced bialgebra morphism $\alpha^* \colon \kofg \to \kofg$.  By Corollary ~\ref{cor:twist} (and Remark ~\ref{rk:R}) we have a quasi-triangular Hom-bialgebra
\[
\kofg_{\alpha^*}^{(n)} = (\kofg,\mu_{\alpha^*},\Delta_{\alpha^*},\alpha^*,e,\rn)
\]
for each $n \geq 1$.  As a function $G \times G \to \bk$, the twisted quasi-triangular structure $\rn$ is given by
\[
\rn(u,v) = \left\{((\alpha^*)^n \otimes (\alpha^*)^n)(R)\right\}(u,v) = R\left(\alpha^n(u),\alpha^n(v)\right)
\]
for $u,v \in G$.  This is the function \eqref{rnfunction} with $\alpha$ and $\alpha^{-1}$ interchanged.
\qed
\end{example}

\begin{example}[\textbf{Hom-quantum enveloping algebras}]
\label{ex:env}
Let us first recall Drinfel'd's quantum enveloping algebra $\uhg$ \cite[section 13]{dri87}, using the notations in \cite[section 3.3]{majid}.  Another exposition of $\uhg$ is given in \cite[XVII]{kassel}.  We refer the reader to \cite{hum,jac} for discussion of semi-simple Lie algebras and to \cite{am,bou,mat} for basics of topological algebras over the power series algebra $\ch$.

Let $\fg$ be a finite dimensional complex semi-simple Lie algebra, $A = (a_{ij})_{1\leq i,j \leq n}$ be its Cartan matrix, $\{\beta_i\}_{1 \leq i \leq n}$ be a system of positive simple roots, and $d_i = (\beta_i,\beta_i)/2$ be its root lengths, where $(-,-)$ is the inverse of the Killing form.  Define the $q$-symbols:
\begin{equation}
\label{qsymbols}
q_i = e^{hd_i/2}, \quad [m]_{q_i} = \frac{q_i^m - q_i^{-m}}{q_i - q_i^{-1}}, \quad
\begin{bmatrix}m\\r\end{bmatrix}_{q_i} = \frac{[m]_{q_i}!}{[r]_{q_i}![m-r]_{q_i}!},
\end{equation}
where $[r]_{q_i}! = [r]_{q_i} [r-1]_{q_i} \cdots [1]_{q_i}$ and $[0]_{q_i}! = 1$.

The quantum enveloping algebra $\uhg$ is defined as the topological $\ch$-algebra that is topologically generated by the set $\{H_i,X_{\pm i}\}_{1 \leq i \leq n}$ of $3n$ generators with the following relations:
\begin{equation}
\label{urelations}
[H_i,H_j] = 0, \quad [H_i,X_{\pm j}] = \pm a_{ij}X_{\pm j}, \quad [X_{+i},X_{-j}] = \delta_{ij}\frac{q_i^{H_i} - q_i^{-H_i}}{q_i - q_i^{-1}},\\
\end{equation}
and if $i \not= j$,
\begin{equation}
\label{ij}
\sum_{k=0}^{1-a_{ij}} (-1)^k\begin{bmatrix}1-a_{ij}\\k\end{bmatrix}_{q_i} X_{\pm i}^{1-a_{ij}-k}X_{\pm j} X_{\pm i}^k = 0.
\end{equation}
The comultiplication $\Delta$ in $\uhg$ is determined by
\begin{equation}
\label{udelta}
\Delta(H_i) = H_i \otimes 1 + 1 \otimes H_i, \quad
\Delta(X_{\pm i}) = X_{\pm i} \otimes q_i^{H_i/2} + q_i^{-H_i/2} \otimes X_{\pm i}
\end{equation}
for $1 \leq i \leq n$.  The bialgebra $\uhg$ becomes a quasi-triangular bialgebra when equipped with the quasi-triangular structure
\begin{equation}
\label{ugr}
R = \sum_{\ba \in \bN^n} \left\{\exp h\left[\frac{1}{2}t_0 + \frac{1}{4}(H_{\ba} \otimes 1 - 1 \otimes H_{\ba})\right]\right\}P_{\ba},
\end{equation}
where $\bN$ is the set of non-negative integers, $H_{\ba} = \sum_{i=1}^n a_iH_i$ for $\ba = (a_1, \ldots , a_n) \in \bN^n$, and $t_0 = \sum_{i,j} (DA)_{ij}^{-1} H_i \otimes H_j$ with $D = \diag(d_1, \ldots , d_n)$ (i.e., the diagonal matrix with $d_i$ as its $i$th diagonal entry).  The symbol $P_{\ba}$ denotes a certain polynomial in the variables $u_i = X_{+i} \otimes 1$ and $v_i = 1 \otimes X_{-i}$ that is homogeneous of degree $a_i$ in $u_i$ and $v_i$, and $P_0 = 1 \otimes 1$.  More information about the quasi-triangular structure $R$ \eqref{ugr} can be found in \cite{kr90,ls90,majid,rosso}.

We can get bialgebra morphisms on $\uhg$ as follows.  Let $\bc = (c_1, \ldots , c_n) \in \bC^n$ be any $n$-tuple of complex numbers.  Define $\hc = \sum_{i=1}^n c_iH_i$ and the $\ch$-linear map $\alphac \colon \uhg \to \uhg$ by
\begin{equation}
\label{alphacdef}
\alphac(u) = e^{h\hc}ue^{-h\hc}
\end{equation}
for $u \in \uhg$.  Then $\alphac$ is clearly an algebra automorphism.  To see that $\alphac$ is compatible with $\Delta$, first note that
\begin{equation}
\label{alphaH}
\alphac(H_j) = H_j
\end{equation}
for all $j$ because $H_i$ commutes with $H_j$ by the first relation in \eqref{urelations}.  So \eqref{udelta} implies that $\alphac^{\otimes 2} \circ \Delta$ and $\Delta \circ \alphac$ coincide when applied to $H_j$.  Also, we have
\begin{equation}
\label{alphaX}
\alphac(X_{\pm j}) = \gamma_j^{\pm 1}X_{\pm j},
\end{equation}
where $\gamma_j = \exp(\sum_{i=1}^n c_ia_{ij})$, by the second relation in \eqref{urelations}.  (More precisely, we are using $e^{chH_i}X_{\pm j}e^{-chH_i} = e^{\pm ca_{ij}}X_{\pm j}$ \cite[p.408 (2.5)]{kassel}, which is a consequence of the second relation in \eqref{urelations}.)  Since $\alphac$ fixes $q_j^{\pm H_j/2} = e^{\pm hd_jH_j/4}$, we infer from \eqref{udelta} that
\[
\alphac^{\otimes 2}(\Delta(X_{\pm j})) = \gamma_j^{\pm 1}\Delta(X_{\pm j}) = \Delta(\alphac(X_{\pm j})).
\]
Therefore, the map $\alphac$ is a bialgebra automorphism on $\uhg$.  Alternatively, one can use \eqref{alphaH} and \eqref{alphaX} as the definition of the map $\alphac$ (on the generators).  Then one checks directly that $\alphac$ preserves the relations \eqref{urelations} and \eqref{ij} and is compatible with the comultiplication \eqref{udelta}.

By Theorem ~\ref{thm:twist1} and Corollary ~\ref{cor:twist}, for each $n$-tuple $\bc \in \bC^n$ and each integer $t \geq 0$, we have a quasi-triangular Hom-bialgebra
\[
\uhg_\alpha^{(t)} = (\uhg,\mu_\alpha,\Delta_\alpha,\alpha,1,\rt)
\]
with $\alpha = \alphac$ \eqref{alphacdef}.  The twisted operations are given by
\[
\begin{split}
\mu_\alpha(u,v) &= e^{h\hc}uve^{-h\hc},\\
\Delta_\alpha(H_j) &= \Delta(H_j), \quad \Delta_\alpha(X_{\pm j}) = \gamma_j^{\pm 1}\Delta(X_{\pm j}) = e^{\pm\sum_{i=1}^n c_ia_{ij}}\Delta(X_{\pm j}).
\end{split}
\]
The twisted quasi-triangular structure $\rt$ is $(\alpha^t \otimes \alpha^t)(R)$, where $\alpha^0 = Id$.  To make it more explicit, note that $\alpha(H_{\ba}) = H_{\ba}$ and $\alpha^{\otimes 2}(t_0) = t_0$ because each $H_j$ is fixed by $\alpha$ \eqref{alphaH}.  So the entire exponential term in $R$ \eqref{ugr} is fixed by $\alpha^{\otimes 2}$.  For the polynomial $P_{\ba}$, let us write it as $P_{\ba}(u_1,\ldots,u_n,v_1,\ldots,v_n)$.  Since
\[
\alpha^{\otimes 2}(u_j) = \alpha(X_{+j}) \otimes 1 = \gamma_j (X_{+j} \otimes 1) = \gamma_j u_j
\]
and similarly $\alpha^{\otimes 2}(v_j) = \gamma_j^{-1} v_j$, we have
\[
(\alpha^t \otimes \alpha^t)P_{\ba}(u_1,\ldots,u_n,v_1,\ldots,v_n) = P_{\ba}(\gamma_1^tu_1, \ldots , \gamma_n^tu_n, \gamma_1^{-t}v_1, \ldots, \gamma_n^{-t}v_n).
\]
Therefore, we have
\[
\rt = \sum_{\ba \in \bN^n} \left\{\exp h\left[\frac{1}{2}t_0 + \frac{1}{4}(H_{\ba} \otimes 1 - 1 \otimes H_{\ba})\right]\right\}P_{\ba}(\gamma_1^tu_1, \ldots , \gamma_n^tu_n, \gamma_1^{-t}v_1, \ldots, \gamma_n^{-t}v_n),
\]
where $\gamma_j^{\pm t}  = \exp(\pm t\sum_{i=1}^n c_ia_{ij})$.
\qed
\end{example}

\begin{example}[\textbf{Hom-quantum enveloping algebra of $sl_2$}]
\label{ex:sl2}
Let us examine the special case $\fg = sl_2$ of the previous example.  The bialgebra $\uhsl$ was first studied in \cite{kr,skl3}, and its quasi-triangular structure \eqref{slR} was given in \cite[p.816]{dri87}.  Using the notations of the previous example with $\fg = sl_2$, we have $n=1$, $a_{11} = 2$, $d_1 = 1$, and $\uhsl$ is the topological $\ch$-algebra generated by $\{H,X_{\pm}\}$ with relations \eqref{urelations}, where $q_1 = q = e^{h/2}$.  The relations \eqref{ij} are empty for $\fg = sl_2$.  The quasi-triangular structure is
\begin{equation}
\label{slR}
R = \sum_{a \geq 0} \frac{(q - q^{-1})^a}{[a]_q!} q^{-a(a+1)/2} \left\{\exp \frac{h}{4}\left[H \otimes H + a(H \otimes 1 - 1 \otimes H)\right]\right\}(X_+^a \otimes X_-^a).
\end{equation}
As in the previous example, given any complex number $c$, we have a bialgebra automorphism $\alpha = \alpha_c \colon \uhsl \to \uhsl$ defined as
\begin{equation}
\label{alphasl2}
\alpha(u) = e^{chH}ue^{-chH}
\end{equation}
for $u \in \uhsl$.

By Theorem ~\ref{thm:twist1} we have a quasi-triangular Hom-bialgebra
\[
\uhsl_\alpha = (\uhsl,\mu_\alpha,\Delta_\alpha,\alpha,1,R).
\]
Moreover, $R$ \eqref{slR} is $\alpha$-invariant, i.e., $\alpha^{\otimes 2}(R) = R$.  This is an immediate consequence of $\alpha(H) = H$ \eqref{alphaH} and $\alpha(X_{\pm}) = e^{\pm 2c}X_{\pm}$ \eqref{alphaX}.  Quasi-triangular Hom-bialgebras with $\alpha$-invariant $R$ play a major role in Theorem ~\ref{thm:hybe} below.  We will revisit this example in section ~\ref{sec:module}.\qed
\end{example}

\section{Solutions of the HYBE from quasi-triangular Hom-bialgebras}
\label{sec:hybe}

In this section, we extend the relationship between the QYBE \eqref{qybe} and the YBE \eqref{ybe}, as discussed in the introduction, to the Hom setting.  Let us first recall the Hom version of the YBE.

\begin{definition}[\cite{yau5}]
\label{def:hybe}
Let $V$ be a $\bk$-module and $\alpha \colon V \to V$ be a linear map.  The \textbf{Hom-Yang-Baxter equation} (HYBE) is defined as
\begin{equation}
\label{hybe}
(\alpha \otimes B) \circ (B \otimes \alpha) \circ (\alpha \otimes B) = (B \otimes \alpha) \circ (\alpha \otimes B) \circ (B \otimes \alpha),
\end{equation}
where $B \colon V^{\otimes 2} \to V^{\otimes 2}$ is a bilinear map that commutes with $\alpha^{\otimes 2}$.  In this case, we say that $B$ is a solution of the HYBE for $(V,\alpha)$.
\end{definition}
The YBE \eqref{ybe} is the special case of the HYBE \eqref{hybe} in which $\alpha = Id$.

As in the classical case, solutions of the HYBE are closely related to the braid relations and braid group representations \cite{artin2,artin}.  Indeed, suppose that $B \colon V^{\otimes 2} \to V^{\otimes 2}$ is a solution of the HYBE for $(V,\alpha)$.  Then for $n \geq 3$ and $1 \leq i \leq n - 1$, the operators
\[
B_i = \alpha^{\otimes(i-1)} \otimes B \otimes \alpha^{\otimes (n - i - 1)} \colon V^{\otimes n} \to V^{\otimes n}
\]
satisfy the braid relations
\[
B_i B_j = B_j B_i \quad \text{ if $|i - j| > 1$} \quad \text{ and } \quad
B_i B_{i+1} B_i = B_{i+1} B_i B_{i+1}.
\]
In particular, if $\alpha$ and $B$ are both invertible, then so are the operators $B_i$.  In this case, there is a corresponding representation of the braid group on $V^{\otimes n}$ \cite[Theorem 1.4]{yau5}.  Many examples of solutions of the HYBE can be found in \cite{yau5,yau6}.

We will show that every quasi-triangular Hom-bialgebra (in which $R$ is fixed by $\alpha^{\otimes 2}$) gives rise to many solutions of the HYBE via its modules.  To make this precise, we need a suitable notion of modules over a Hom-associative algebra.

\begin{definition}
\label{def:module}
\begin{enumerate}
\item
A \textbf{Hom-module} is a pair $(V,\alpha)$ consisting of a $\bk$-module $V$ and a linear map $\alpha$.
\item
Let $(A,\mu,\alpha_A)$ be a Hom-associative algebra (Definition ~\ref{def:homas}).  By an \textbf{$A$-module} we mean a Hom-module $(M,\alpha_M)$ together with a linear map $\lambda \colon A \otimes M \to M$ such that
\begin{equation}
\label{moduleaxioms}
(ab)\alpha_M(x) = \alpha_A(a)(bx) \quad\text{and}\quad \alpha_M(ax) = \alpha_A(a)\alpha_M(x)
\end{equation}
for $a,b \in A$ and $x \in M$, where $\lambda(a,x)$ is abbreviated to $ax$.
\end{enumerate}
\end{definition}

Note that a slightly different notion of a module over a Hom-associative algebra was defined in \cite{ms2}.  The difference with Definition ~\ref{def:module} is that in \cite{ms2}, the second condition in \eqref{moduleaxioms} is not required.

\begin{example}
\label{ex:module}
\begin{enumerate}
\item
A Hom-associative algebra $(A,\mu,\alpha)$ is an $A$-module with structure map $\lambda = \mu$.  In this case, the axioms \eqref{moduleaxioms} are exactly the Hom-associativity and the multiplicativity of $\alpha$ in $A$.
\item
Let $(A,\mu)$ be an associative algebra, $M$ be an $A$-module (in the usual sense) with structure map $\lambda$, $\alpha_A \colon A \to A$ be an algebra morphism, and $\alpha_M \colon M \to M$ be a linear map.  Suppose that $\alpha_M \circ \lambda = \lambda \circ (\alpha_A \otimes \alpha_M)$.  This is the case, for example, if $\alpha_A = Id$ and $\alpha_M$ is an $A$-module morphism.  Define the twisted action $\lambda_\alpha = \alpha_M \circ \lambda$.  Then it is easy to check that  $(M,\alpha_M)$ becomes a module over the Hom-associative algebra $A_\alpha = (A,\mu_\alpha = \alpha_A \circ \mu,\alpha_A)$ (Example ~\ref{ex:homas}) with structure map $\lambda_\alpha$.\qed
\end{enumerate}
\end{example}

In a quasi-triangular Hom-bialgebra (Definition ~\ref{def:qthb}), we say that the element $R$ is \textbf{$\alpha$-invariant} if $\alpha^{\otimes 2}(R) = R$.  Some examples of $\alpha$-invariant $R$ were given in Example ~\ref{ex:sl2}.  When $R$ is $\alpha$-invariant, the two versions of the QHYBE (\eqref{qhybe} and \eqref{qhybe'}) coincide, as was noted in Remark ~\ref{rk:alphainv}.

We can now describe the relationship between quasi-triangular Hom-bialgebras and the HYBE \eqref{hybe}.


\begin{theorem}
\label{thm:hybe}
Let $(A,\mu,\Delta,\alpha,c,R)$ be a quasi-triangular Hom-bialgebra in which $R$ is $\alpha$-invariant and $(M,\alpha_M)$ be an $A$-module.  Then the operator
\[
B = \tau \circ R \colon M^{\otimes 2} \to M^{\otimes 2}
\]
is a solution of the HYBE for $(M,\alpha_M)$.
\end{theorem}

\begin{proof}
To simplify the typography, we will omit the subscripts in $\alpha_A$ and $\alpha_M$.  Write $R = \sum s_i \otimes t_i$.  Then the $\alpha$-invariance of $R$ means that
\begin{equation}
\label{alphainvariance}
\sum s_i \otimes t_i = \sum \alpha(s_i) \otimes \alpha(t_i).
\end{equation}
Recall that $\tau$ denotes the twist isomorphism.  So the map $B = \tau \circ R$ is given by
\[
B(v \otimes w) = \sum t_iw \otimes s_iv
\]
for $v,w \in M$.  That $B$ commutes with $\alpha^{\otimes 2}$ follows from the $\alpha$-invariance of $R$ and the second axiom in \eqref{moduleaxioms}.

It remains to check that $B = \tau \circ R$ satisfies the HYBE \eqref{hybe}.  Note that the $\alpha$-invariance of $R$ \eqref{alphainvariance} and the computation \eqref{qlhs} imply that the QHYBE \eqref{qhybe} now takes the form
\begin{equation}
\label{q}
\sum s_js_i \otimes t_js_k \otimes t_it_k = (R_{12}R_{13})R_{23} = R_{23}(R_{13}R_{12}) = \sum s_js_i \otimes s_kt_i \otimes t_kt_j.
\end{equation}
Let $x$ denote a typical generator $u \otimes v \otimes w \in M^{\otimes 3}$. Using the $\alpha$-invariance of $R$ \eqref{alphainvariance} and the module axioms \eqref{moduleaxioms}, a direct computation gives:
\begin{subequations}
\allowdisplaybreaks
\begin{align*}
(B \otimes \alpha)(\alpha \otimes B)(B \otimes \alpha)(x)
&= t_k(t_j\alpha(w)) \otimes s_k\alpha(t_iv) \otimes \alpha(s_j(s_iu))\\
&= t_k(t_j\alpha(w)) \otimes s_k(\alpha(t_i)\alpha(v)) \otimes \alpha(s_j)(\alpha(s_i)\alpha(u))\\
&= \alpha(t_k)(t_j\alpha(w)) \otimes \alpha(s_k)(t_i\alpha(v)) \otimes \alpha(s_j)(s_i\alpha(u))\\
&= (t_kt_j)\alpha^2(w) \otimes (s_kt_i)\alpha^2(v) \otimes (s_js_i)\alpha^2(u)\\
&= (t_it_k)\alpha^2(w) \otimes (t_js_k)\alpha^2(v) \otimes (s_js_i)\alpha^2(u) \quad\text{by \eqref{q}}\\
&= \alpha(t_i)(t_k\alpha(w)) \otimes \alpha(t_j)(s_k\alpha(v)) \otimes \alpha(s_j)(s_i\alpha(u))\\
&= \alpha(t_i)(\alpha(t_k)\alpha(w)) \otimes t_j(\alpha(s_k)\alpha(v)) \otimes s_j(s_i\alpha(u))\\
&= \alpha(t_i(t_kw)) \otimes t_j\alpha(s_kv) \otimes s_j(s_i\alpha(u))\\
&= (\alpha \otimes B)(B \otimes \alpha)(\alpha \otimes B)(x).
\end{align*}
\end{subequations}
This proves that $B = \tau \circ R$ is a solution of the HYBE for $(M,\alpha_M)$.
\end{proof}

\begin{corollary}
\label{cor:hybe}
Let $(A,\mu,\Delta,R)$ be a quasi-triangular bialgebra, $M$ be an $A$-module with structure map $\lambda$, $\alpha_A \colon A \to A$ be a bialgebra morphism such that $\alpha_A^{\otimes 2}(R) = R$, and $\alpha_M \colon M \to M$ be a linear map such that $\alpha_M \circ \lambda = \lambda \circ (\alpha_A \otimes \alpha_M)$.  Then the operator $B_\alpha \colon M^{\otimes 2} \to M^{\otimes 2}$ defined by
\begin{equation}
\label{Balpha}
B_\alpha(v \otimes w) = \sum \alpha_M(\lambda(t_i,w)) \otimes \alpha_M(\lambda(s_i,v))
\end{equation}
for $v,w \in M$, where $R = \sum s_i \otimes t_i$, is a solution of the HYBE for $(M,\alpha_M)$.
\end{corollary}

\begin{proof}
By Theorem ~\ref{thm:twist1} $A_\alpha = (A,\mu_\alpha,\Delta_\alpha,\alpha,1,R)$ is a quasi-triangular Hom-bialgebra, and $(M,\alpha_M)$ is an $A_\alpha$-module (Example ~\ref{ex:module}) with structure map $\lambda_\alpha = \alpha_M \circ \lambda$.  Therefore, Theorem ~\ref{thm:hybe} implies that there is a solution of the HYBE for $(M,\alpha_M)$ of the form
\[
B_\alpha(v \otimes w) = \sum \lambda_\alpha(t_i,w) \otimes \lambda_\alpha(s_i,v)
= \sum \alpha_M(\lambda(t_i,w)) \otimes \alpha_M(\lambda(s_i,v)),
\]
as was to be shown.
\end{proof}

The following result is the special case of the previous Corollary with $\alpha_A = Id$.

\begin{corollary}
\label{cor2:hybe}
Let $(A,\mu,\Delta,R)$ be a quasi-triangular bialgebra, $M$ be an $A$-module, and $\alpha_M$ be an $A$-module morphism.  Then the operator $B_\alpha \colon M^{\otimes 2} \to M^{\otimes 2}$ defined in \eqref{Balpha} is a solution of the HYBE for $(M,\alpha_M)$.
\end{corollary}

\section{Modules over $\uhsl_\alpha$}
\label{sec:module}

In this section, we illustrate the results in the previous section with certain modules over the quasi-triangular Hom-bialgebra $\uhsl_\alpha$, which was discussed in Example ~\ref{ex:sl2}.  We use the same notations as in Examples ~\ref{ex:env} and ~\ref{ex:sl2}.  In particular, $\uhsl$ is the topological $\ch$-algebra generated by $\{H,X_{\pm}\}$ with relations \eqref{urelations}, where $q_1 = q = e^{h/2}$, and its comultiplication is defined as in \eqref{udelta}.  It becomes a quasi-triangular Hom-bialgebra when equipped with the quasi-triangular structure $R$ \eqref{slR}.

Fix a complex number $c$, and let $\alpha \colon \uhsl \to \uhsl$ be the bialgebra automorphism defined by $\alpha(H) = H$ and $\alpha(X_{\pm}) = \gamma^{\pm 1}X_{\pm}$, where $\gamma = e^{2c}$.  Equivalently, $\alpha$ is the inner automorphism \eqref{alphasl2} induced by $e^{chH}$.  Then $\uhsl_\alpha$ is the quasi-triangular Hom-bialgebra obtained from $\uhsl$ by twisting its (co)multiplication along $\alpha$ (Theorem ~\ref{thm:twist1}).  Moreover, the element $R$ \eqref{slR} is $\alpha$-invariant, i.e., $(\alpha \otimes \alpha)(R) = R$.

Fix a non-negative integer $n$.  Let $\vn$ be the free $\ch$-module with a basis $\{v_i\}_{0 \leq i \leq n}$.  Then $\vn$ becomes a (topological) $\uhsl$-module via the map $\rho \colon \uhsl \otimes \vn \to \vn$ determined by
\begin{equation}
\label{rhov}
\rho(X_+,v_i) = [n+1-i]_q v_{i-1},\quad
\rho(X_-,v_i) = [i+1]_q v_{i+1},\quad
\rho(H,v_i) = (n-2i)v_i
\end{equation}
for $0 \leq i \leq n$.  In \eqref{rhov} we set $v_{-1} = 0 = v_{n+1}$, and $[m]_q$ is defined as in \eqref{qsymbols}.  See, for example, \cite[XVII.4]{kassel} and \cite{dri87,kr,skl3}.  We will apply Corollary ~\ref{cor:hybe} to the $\uhsl$-module $\vn$.

Consider the $\ch$-linear automorphism $\alpha \colon \vn \to \vn$ defined by
\begin{equation}
\label{alphavn}
\alpha(v_i) = \gamma^{-i}v_i
\end{equation}
for $0 \leq i \leq n$.

\begin{lemma}
\label{lem:vn}
We have
\begin{equation}
\label{rhoalphavn}
\alpha \circ \rho = \rho \circ (\alpha \otimes \alpha)
\end{equation}
as maps $\uhsl \otimes \vn \to \vn$.
\end{lemma}

\begin{proof}
It suffices to check \eqref{rhoalphavn} on the elements $X_{\pm} \otimes v_i$ and $H \otimes v_i$.  When applied to $H \otimes v_i$, both sides of \eqref{rhoalphavn} are equal to $(n-2i)\gamma^{-i}v_i$.  On the other hand, we have
\[
\begin{split}
\alpha(\rho(X_+,v_i)) &= \alpha([n+1-i]_q v_{i-1})\\
&= [n+1-i]_q \gamma^{-i+1}v_{i-1}\\
&= \rho(\gamma X_+,\gamma^{-i}v_i)\\
&= \rho(\alpha(X_+),\alpha(v_i)).
\end{split}
\]
A similar computation shows that both sides of \eqref{rhoalphavn}, when applied to $X_-\otimes v_i$, are equal to $[i+1]_q\gamma^{-i-1}v_{i+1}$.
\end{proof}

\begin{proposition}
\label{vnmodule}
The map $\rho_\alpha = \alpha \circ \rho$ \eqref{rhoalphavn} gives $(\vn,\alpha)$ the structure of a $\uhsl_\alpha$-module.
\end{proposition}

\begin{proof}
This is an immediate consequence of Lemma ~\ref{lem:vn} and Example ~\ref{ex:module} (part (2)).  
\end{proof}

Moreover, by Lemma ~\ref{lem:vn} and Corollary ~\ref{cor:hybe}, there is a solution of the HYBE for $(\vn,\alpha)$ of the form \eqref{Balpha}:
\begin{equation}
\label{BalphaVn}
B_\alpha = \alpha^{\otimes 2} \circ \tau \circ R \colon \vn^{\otimes 2} \to \vn^{\otimes 2},
\end{equation}
where $R$ \eqref{slR} acts on $\vn^{\otimes 2}$ via the original $\uhsl$-module structure $\rho$.  Let us write down $B_\alpha$ explicitly for the case $\vone$.

\begin{proposition}
\label{BalphaV1}
With respect to the basis $\{v_0 \otimes v_0, v_0 \otimes v_1, v_1 \otimes v_0, v_1 \otimes v_1\}$ of $\vone^{\otimes 2}$, the solution $B_\alpha = \alpha^{\otimes 2} \circ \tau \circ R$ of the HYBE for $(\vone,\alpha)$ is given by the matrix
\begin{equation}
\label{V1}
B_\alpha = q^{-\frac{1}{2}}\Bmatrix,
\end{equation}
where $q = e^{h/2}$.
\end{proposition}

\begin{proof}
Let us first compute the action of $R$ \eqref{slR} on $\vone^{\otimes 2}$.  It follows from the definition ~\eqref{rhov} that, when $n = 1$, both $X_+^a$ and $X_-^a$ act trivially on $\vone$ for $a \geq 2$, and
\begin{equation}
\label{Hv}
Hv_i = \rho(H,v_i) = (-1)^i v_i
\end{equation}
for $i = 0,1$.  Thus, only the first two terms in $R$ \eqref{slR} (corresponding to $a = 0,1$) can act non-trivially on $\vone^{\otimes 2}$.  Since $q = e^{h/2}$, these two terms are
\begin{equation}
\label{R'}
R' = q^{\frac{1}{2}(H \otimes H)} + (1 - q^{-2})q^{\frac{1}{2}(H \otimes H + H \otimes 1 - 1 \otimes H)} (X_+ \otimes X_-).
\end{equation}
Consider the action of the first term in $R'$.  It follows from \eqref{Hv} that $H^mv_0 = v_0$ and $H^mv_1 = (-1)^mv_1$ for $m \geq 0$.  Thus, we have
\begin{equation}
\label{R'2}
q^{\frac{1}{2}(H \otimes H)} = e^{\frac{h}{4}(H\otimes H)} \colon
\begin{cases}
v_0 \otimes v_0 & \mapsto~ q^{\frac{1}{2}}v_0 \otimes v_0,\\
v_0 \otimes v_1 & \mapsto~ q^{-\frac{1}{2}}v_0 \otimes v_1,\\
v_1 \otimes v_0 & \mapsto~ q^{-\frac{1}{2}}v_1 \otimes v_0,\\
v_1 \otimes v_1 & \mapsto~ q^{\frac{1}{2}}v_1 \otimes v_1.
\end{cases}
\end{equation}
For the second term in $R'$, note that $X_+ \otimes X_-$ only acts non-trivially on $v_1 \otimes v_0$ among the four basis elements $\{v_i \otimes v_j\}_{0\leq i,j\leq 1}$.  We have
\begin{equation}
\label{R'1}
(X_+ \otimes X_-)(v_1 \otimes v_0) = v_0 \otimes v_1
\end{equation}
and
\begin{equation}
\label{R'3}
q^{\frac{1}{2}(H \otimes H + H \otimes 1 - 1 \otimes H)} = e^{\frac{h}{4}(H \otimes H + H \otimes 1 - 1 \otimes H)} \colon v_0 \otimes v_1 \mapsto q^{\frac{1}{2}} v_0 \otimes v_1.
\end{equation}
The result now follows from \eqref{alphavn}, \eqref{BalphaVn}, and \eqref{R'} - \eqref{R'3}.
\end{proof}

Regard $\gamma = e^{2c}$ as a parameter, where $c$ runs through the complex numbers.  The operators $B_\alpha$ \eqref{V1} thus form a $1$-parameter family of deformations of
\[
B = q^{-\frac{1}{2}}\Brmatrix,
\]
which is a solution of the YBE \eqref{ybe} for $\vone$.


\end{document}